\include{Macros}
\documentclass[twocolumn,12pt]{autart}    % Enable this line and disable the 
\theoremstyle{plain}
\usepackage[font=small, labelfont=bf]{caption}
\usepackage[normalem]{ulem}
\usepackage{mathdesign}
\usepackage{draftcopy}
\usepackage{color}
\usepackage{graphicx}
\usepackage{epsfig}
\usepackage{amsmath}
\usepackage{amssymb}
\usepackage{latexsym}
\usepackage{amsfonts}
\usepackage{algorithm}
\usepackage{algorithmic}
\usepackage{ulem}
\usepackage{cite}
\usepackage{subfigure}
\usepackage{graphics}
\usepackage{latexsym}
\usepackage{graphicx}
\usepackage{makeidx}
\usepackage{amsmath}
\usepackage{amssymb}
\usepackage{psfrag}
 \newtheorem{theorem}{Theorem}
 \newtheorem{lemma}{Lemma}
 \newtheorem{proposition}{Proposition}
  \newtheorem{corollary}{Corollary}
   \newtheorem{definition}{Definition}
 \newtheorem{remark}{Remark}
 \newtheorem{assumption}{Assumption}
 \newtheorem{example}{Example}
 \newtheorem{step}{Step}
  
  \newtheorem{Steps}{Step}

 \newtheorem{problem}{Problem}

\newcommand{\vect}[1]{\boldsymbol{#1}}

\renewcommand{\emph}{\textit}

\newenvironment{proof}[1][Proof]{\begin{trivlist} \item[\hskip \labelsep {\bfseries #1}]}{\end{trivlist}}
\begin{document}
\begin{frontmatter}
\title{On minimal realisations of dynamical structure functions}\thanks{This paper has not been published in any conference, some preliminary results in Section~\ref{sec:special} have been published in \cite{yetac}. Ye Yuan and Jorge Gon\c{c}alves acknowledge the support from EPSRC through EP/I03210X/1, EP/G066477/1.
}
\author[1]{Ye Yuan}, \author[1]{Keith Glover} and \author[1,2]{Jorge Gon\c{c}alves}
\maketitle
\address[1]{Control Group, Department of Engineering, University of Cambridge, UK.} 
\address[2]{Luxembourg Centre for Systems Biomedicine, Luxembourg.}      

\begin{abstract}
Motivated by the fact that transfer functions do not contain structural information about networks, dynamical structure functions were introduced to capture causal relationships between measured nodes in networks. From the dynamical structure functions, a) we show that the actual number of hidden states can be larger than the number of hidden states estimated from the corresponding transfer function; b) we can obtain partial information about the true state-space equation, which cannot in general be obtained from the transfer function. Based on these properties, this paper proposes algorithms to find minimal realisations for a given dynamical structure function. This helps to estimate the minimal number of hidden states, to better understand the complexity of the network, and to identify potential targets for new measurements.
\end{abstract}
\end{frontmatter}

\section{INTRODUCTION}
Networks have received increasing attention in the last
decade. In our ``information-rich'' world, questions pertaining to network
reconstruction and network analysis have become crucial for the
understanding of complex systems. In particular, the analysis of molecular
networks has gained significant interest due to the recent explosion
of publicly available high-throughput biological data. Another example are social networks, which are social structures made up of individuals, the nodes, tied by one or more specific types of interdependencies, the edges (e.g. friendship).  
In this context, identifying and analysing network structures from measured data become key questions. 

To mathematically represent networks, we use the standard graph-theoretical notation $\mathcal{G}=\left(\mathcal{V},\mathcal{E},A\right)$, where $\mathcal{V}=\left\{\nu_{1},\ldots,\nu_{n}\right\}$ is the set of nodes, $\mathcal{E}\subset\mathcal{V}\times\mathcal{V}$ is the set of edges, and $A=\left\{A[i,j]\right\}_{i,j=1,\ldots,n}$ is the corresponding $n$ by $n$ weighted adjacency matrix, with $A[i,j]\neq0$ when there is a link from $j$ to $i$, and $A[i,j]=0$ when there is no link from $j$ to $i$. 
In the classic state-space form, we usually write
\begin{align}
\dot{x}(t)&={A}{x}(t)+{B}{u}(t) \nonumber\\
y(t)&=Cx(t)+Du(t) \label{eq:ss}
\end{align}
$x(t)\in\mathbb{R}^n$ is the state vector containing the state (normally physical quantity) of the system. ${A}\in\mathbb{R}^{n\times n}$ is the weighted adjacency matrix reflecting the direct causal relations between the state variables, ${B}\in\mathbb{R}^{n\times m}$, and ${u}(t)\in\mathbb{R}^m$ is a vector of $m$ inputs.  

This work assumes that $p<n$ states are measured. Without loss of generality, the output equation can be written as ${y}(t) = {C}{x}(t)$, where ${C}=[I_p \ \ 0]$, ${I_p}$ is the $p\times p$ identity matrix, and ${0}$ is the $p \times (n-p)$ matrix of zeros.  Hence, the first $p$ elements of the state vector $x$ are exactly the measured variables in the system, and the remaining $(n-p)$ state variables are unmeasured ``hidden" states. The zero structure of the ${A}$ and ${B}$ matrices exactly describe the structure of the network, and the values of these matrices encode the dynamics of the system.

Finding the matrices $A$ and $B$, however, can be a difficult problem in the presence of hidden states ($p<n$). Even with just one hidden state, the realisation problem becomes ill-posed; a transfer function will have many state space realisations, which may suggest entirely different network structures for the system.  This is true even if it is known that the true system is, in fact, a minimal realisation of the identified transfer function.  As a result, failure to explicitly acknowledge the presence and the ambiguity in network structure caused hidden states can lead to a deceptive and erroneous process for network discovery.

Motived by this, we developed a new theory for network inference that reflected the effects of hidden states in a network~\cite{08net_rec}. It introduced a new representation for LTI systems called dynamical structure functions (DSF). DSF capture information at an intermediate level between transfer function and state space
representation (see Figure~\ref{Fig:math}). Specifically, dynamical structure functions not only encode structural information at the measurement level, but also contain some
information about hidden states. In \cite{08net_rec}, we proposed some guidelines for the design of an experimental data-acquisition protocol which allows the collection of data containing sufficient information for the network structure reconstruction problem
to become solvable. Using dynamical structure functions as a means to solve the network reconstruction problem, the following aspects need to be considered:
\begin{figure}
\centering
\includegraphics[scale=0.3]{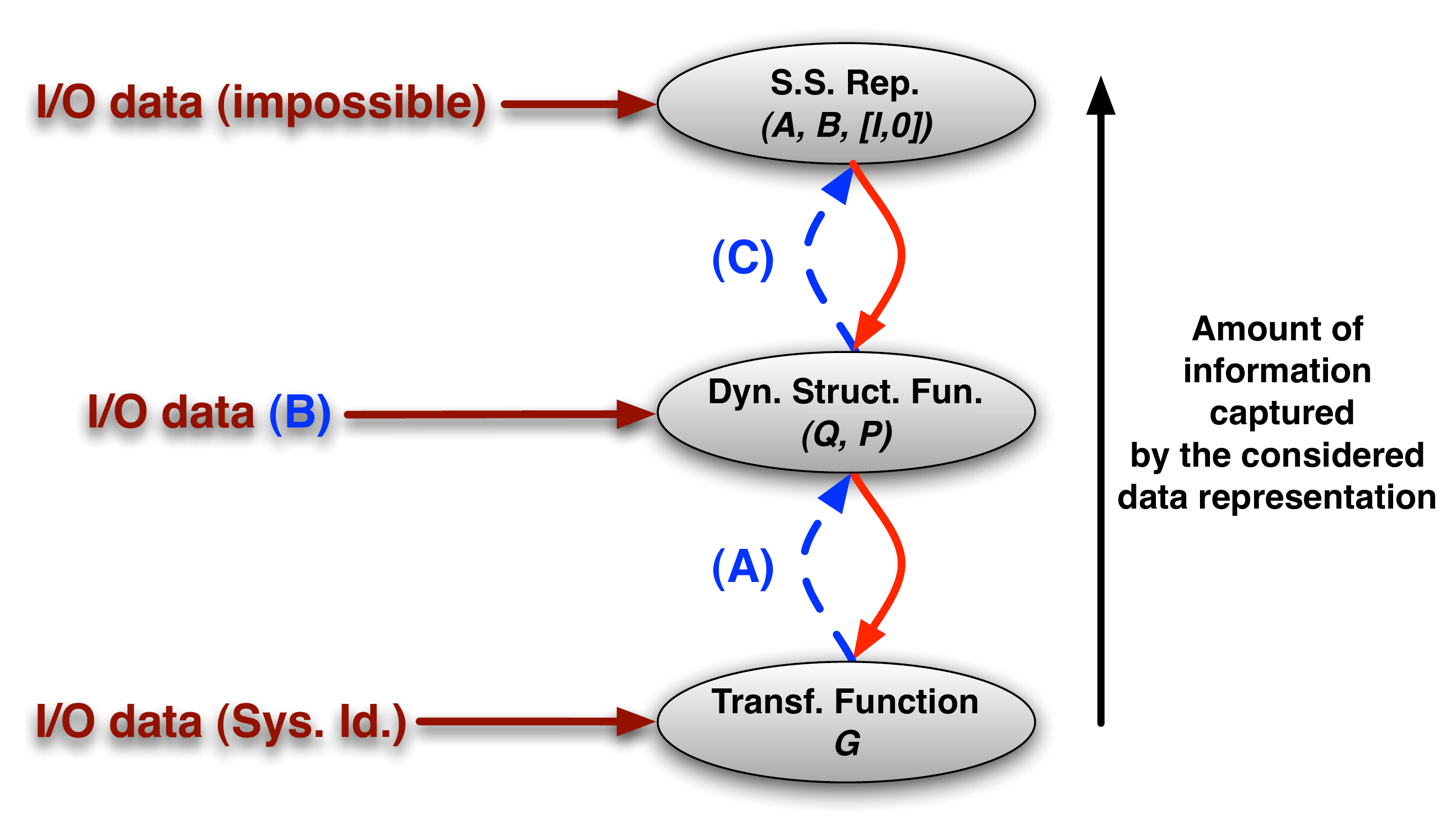}
\caption{Mathematical structure of the network reconstruction problem using dynamical structure functions. Red arrows mean ``uniquely determine'', blue arrows indicate our work.}\label{Fig:math}
\end{figure}

First (see (A) in Figure~\ref{Fig:math}), the properties of a dynamical
structure function and its relationship with the transfer
function associated with the same system were precisely established~\cite{08net_rec}.
%We will show that if
%experiments are performed as explained above: 1) we can not only
%obtain the network between the measured states but also the
%``self-loop'' gain for each measured state; 2) we can still recover
%the true network structure even if the exact value of control
%input is unknown. 

Second (see (B) in Figure~\ref{Fig:math}), an efficient method to reconstruct networks in the presence of noise and nonlinearities was developed~\cite{robust}.
%The method presented in~\cite{robust} relies on the assumption that the conditions for network reconstruction presented above in (A.1) and (A.2) have been met.
In this method, steady-state (resp. time-series data) can be used to reconstruct the Boolean (resp. dynamical network) structure of the system (see~\cite{robust} for more details). 

Third (see (C) in Figure~\ref{Fig:math}), once the dynamical structure function is obtained, an algorithm for constructing a minimal order state-space realisation of such function needs to be developed. This third point is the main contribution of this paper. In an application, this provides an estimate of the complexity of the system by determining the minimal number of hidden states in the system. For example, in the context of biology it helps understand the number of unmeasured molecules in a particular pathway: a low number of hidden states means that most molecules in that pathway have been identified and measured, showing a good understanding of the system; while a large number shows that there are still many unmeasured variables, suggesting that new experiments should be carried out to better characterise that pathway.

For a given dynamical structure function,the major contributions of this paper are: 
\begin{itemize} 
\item[a)]  it explicitly characterises the {\em direct} causal information between measured states and between measured states and inputs;
\item[b)] it introduces a number of new concepts such as hidden observability and controllability; 
\item[c)] it extends the results in \cite{yetac} by considering the minimal realisation problem of more general classes of dynamical structure functions. 
\end{itemize} 

The notation in this paper is standard. For a matrix $A \in \mathbb{R}^{M \times N}$, $A[i,j] \in
\mathbb{R}$ denotes the element in the $i^{th}$ row and $j^{th}$
column, $A[i,:] \in \mathbb{R}^{1 \times N}$ denotes its $i^{th}$
row, $A[:,j] \in \mathbb{R}^{M \times 1}$ denotes its $j^{th}$
column, and $A[i_1:i_2,j_1:j_2]$ 
% \in \mathbb{R}^{(i_2-i_1+1)\times(j_2-j_1+1)}$ 
denotes the submatrix of $A$ defined by the rows $i_1$ to $i_2$ and the columns $j_1$ to $j_2$.
For a column vector $\alpha \in \mathbb{R}^{N\times 1}$, $\alpha[i]$
denotes its $i^{th}$ element. 
%Similarly, for a row vector $\beta \in
%\mathbb{R}^{1\times N}$, $\beta[i]$ denotes its $i^{th}$ element. 
We denote $e_r^T=[0,\ldots,0,1_{r^{th}},0,\ldots,0] \in
\mathbb{R}^{1 \times N}$. Furthermore, 
$I_N$ denotes the identity matrix of size $N$.

\section{DYNAMICAL STRUCTURE FUNCTIONS AND ITS PROPERTIES}\label{sec:systemmodel}
Consider the following linear system (we put a superscript $o$ indicating the original system)
\begin{equation}\label{eq:LTI}
 \begin{array}{cll}
   \left[\begin{array}{c}{\dot{y}}\\{\dot{z}} \end{array}\right]& = &
   \left[\begin{array}{cc}{A}^o_{11}&{A}^o_{12}\\{A}^o_{21}&{A}^o_{22}\end{array}\right]\left[\begin{array}{c}{y}\\{z}\end{array}\right]
   +\left[\begin{array}{c} {B}^o_{1}\\{B}^o_{2}\end{array}\right]{u} \\
   {y} & = &
   \left[\begin{array}{cc}{I}_p&{0}\end{array}\right]\left[\begin{array}{c}{y}\\{z}\end{array}\right],
   \end{array}
\end{equation}
where ${x}=({y},{z}) \in \Bbb R^{n^o}$ is the full
state vector, ${y}\in \Bbb R^p$ is a partial measurement of the
state, ${z}$ are the $n^o-p$ ``hidden'' states, and
${u}\in \Bbb R^m$ is the control input. In this work we restrict our attention to situations where output measurements constitute partial state information, i.e., $p< n^o$. It is well known that the transfer function of this system can be defined by $G^o\triangleq[I_p~0](sI-A^o)^{-1}B^o$. 

\subsection{Definitions of transfer functions and dynamical structure functions}
Dynamical structure functions can be uniquely determined by state-space realisations. It is more involved comparing with the definition of a transfer function \cite{08net_rec}.

Taking the Laplace transforms of the signals in~(\ref{eq:LTI}) yields
\begin{equation}\label{eq:LTIlaplace}
 \begin{array}{lll}
   \left[\begin{array}{c}s{Y}\\s{Z} \end{array}\right]& = & \left[\begin{array}{cc}{A}^o_{11}&{A}^o_{12}\\{A}^o_{21}&{A}^o_{22}\end{array}\right]\left[\begin{array}{c}{Y}\\{Z}\end{array}\right] +\left[\begin{array}{c} {B}^o_{1}\\{B}^o_{2}\end{array}\right]{U}
     \end{array}
\end{equation}
where ${Y}$, ${Z}$, and ${U}$ are the Laplace
transforms of ${y}$, ${z}$, and ${u}$, respectively.
Solving for ${Z}$ gives
$${Z}=\left ( s{I} - {A}^o_{22} \right )^{-1}
{A}^o_{21} {Y} + \left ( s{I} - {A}^o_{22}
\right )^{-1} {B}^o_{2} {U}$$
Substituting this last expression of ${Z}$ into~(\ref{eq:LTIlaplace}) then yields
\begin{equation}
\label{eq:WV} s{Y} = {W}^o {Y} + {V}^o{U}
\end{equation}
where ${W}^o={A}^o_{11} + {A}^o_{12}\left ( s{I} -
  {A}^o_{22} \right )^{-1} {A}^o_{21}$ and
${V}^o={B}^o_{1} +{A}^o_{12}\left ( s{I} - {A}^o_{22} \right
)^{-1} {B}^o_{2}$.

Now, let ${R}^o$ be a
diagonal matrix formed of the diagonal terms of ${W}^o$ on its
diagonal, i.e., ${R}^o=\mbox{diag}\{{W}^o\} =
\mbox{diag}(W^o_{11}, W^o_{22}, ..., W^o_{pp})$. Subtracting ${R}^o{Y}$ from both sides of \eqref{eq:WV}, we obtain:
$$\left ( s{I} - {R}^o \right ) {Y} = \left (
  {W}^o-{R}^o \right ) {Y} + {V}^o {U}$$
Note that ${W}^o-{R}^o$ is a matrix with zeros on its
diagonal. We thus have:
\begin{equation}
  \label{eq:PQ} {Y} =  {QY} + {PU}
\end{equation}
where
\begin{equation}\label{eq:Q}
  {Q} = \left ( s{I} - {R}^o \right )^{-1} \left (
    {W}^o-{R}^o \right )
\end{equation}
and
\begin{equation}\label{eq:P}
  {P}=\left ( s{I}- {R}^o \right )^{-1} {V}^o
\end{equation}
Note that ${Q}$ has zero on the diagonal. Given the system in~(\ref{eq:LTI}), $[{Q},{P}]$ denotes the {\it dynamical structure functions} of the system.

\subsection{Final value properties}
Next we shall explore some properties of $[Q,P]$. One of the most important properties is that the dynamical structure functions capture the {\em direct} causal relations between measured states $y$.
\begin{proposition}\label{th:qp}
Given a dynamical system~(\ref{eq:LTI}) and its associated dynamical structure
  functions $[{Q},{P}]$ with ${R}^o$ constructed as
  explained above (see \eqref{eq:LTI}-\eqref{eq:P}), the following conditions must hold
\begin{align}
  \text{diag}\{{A}^o_{11}\}& = \lim_{s\rightarrow\infty}{R}^o(s)\label{eq:Ds};\\
  {A}^o_{11}-\text{diag}\{{A}^o_{11}\}&=\lim_{s\rightarrow \infty} s {Q}(s)\label{eq:Qs};\\
  {B}^o_1&=\lim_{s\rightarrow \infty} s {P}(s)\label{eq:Ps}.
\end{align}
\end{proposition}
\begin{proof}
See Appendix~\ref{sec:appendixA}.
\end{proof}

\begin{remark}
This Proposition reveals an important property of dynamical structure functions: they encode the direct causal relations between observed variables, i.e., $A^o_{11}[i,j]~,\forall i\neq j$. These relations cannot be revealed by transfer functions. 
\end{remark}

\begin{example} \label{ex1}
Consider a network with the structure depicted in Fig. \ref{fig:ex1}. The linear state-space representation of this network is given by
\begin{equation*} \begin{array}{rcl}
\dot{x} &=& 
\begin{bmatrix}a_{11} & 0 & a_{13} & 0 & 0 \\
			0 & a_{22} & 0 & a_{24} & 0 \\
			0 & a_{32} & a_{33} & 0 & a_{35} \\
			a_{41} & 0 & 0 & a_{44} & 0 \\
			0 & a_{52} & 0 & 0 & a_{55} \end{bmatrix} x + 
\begin{bmatrix}b_{11} & 0 \\ 0 & b_{22} \\ 0 & 0 \\ 0 & 0 \\ 0 & 0 \end{bmatrix} u \\
y &=& \begin{bmatrix} I_3 & 0 \end{bmatrix} x.
\end{array} \end{equation*}
Following the definitions in \eqref{eq:Q} and \eqref{eq:P}, the corresponding dynamical structure functions $[{Q},{P}]$ are
\begin{align*}
Q &= \begin{pmatrix} 0 & 0 & \frac{a_{13}}{s-a_{11}} \\ \frac{a_{24}a_{41}}{(s-a_{22})(s-a_{44})} & 0 & 0 \\ 0 & \frac{a_{35}a_{52} + a_{32}(s-a_{55})}{(s-a_{33})(s-a_{55})} & 0 \end{pmatrix}, \\
P &= \begin{pmatrix} \frac{b_{11}}{s-a_{11}} & 0 \\ 0 & \frac{b_{22}}{s-a_{22}} \\ 0 & 0 \end{pmatrix}.
\end{align*}
\begin{figure}[!] \centering
\includegraphics[width=1\linewidth]{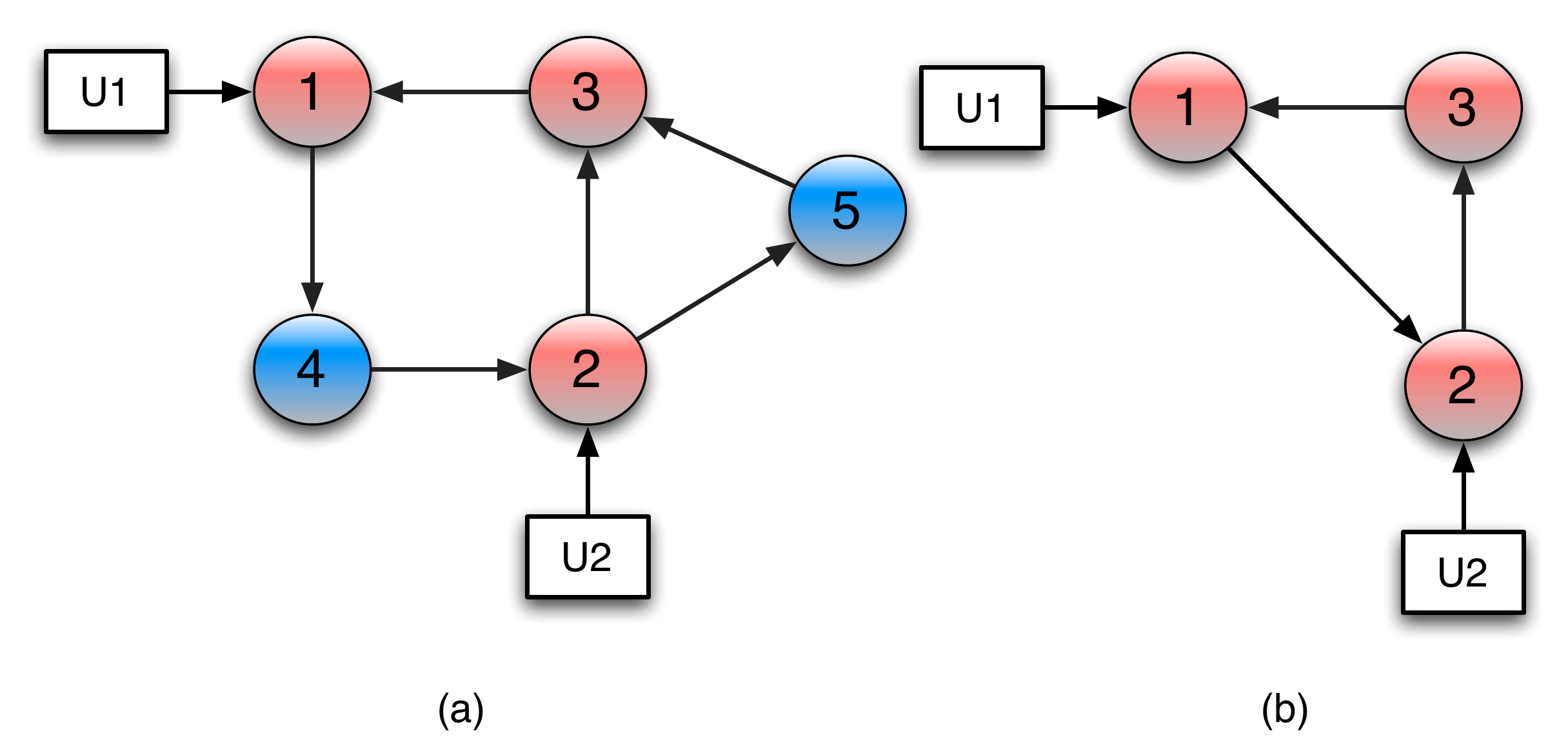}
\caption{\textbf{(a)} An example system with two inputs, three measured states (red states $1$, $2$, and $3$) and two hidden states (blue states $4$ and $5$). \textbf{(b)} The corresponding dynamical structure functions.}
\label{fig:ex1}
\end{figure}
From Proposition~\ref{th:qp}, we can check that:
\begin{align*}
  \lim_{s\rightarrow \infty} s {Q}(s)&=\begin{pmatrix} 0 & 0 & a_{13}  \\
	0 & 0 & 0  \\
	0 & a_{32} & 0 \end{pmatrix} ;\\
\lim_{s\rightarrow \infty} s {P}(s)&=s\begin{pmatrix} \frac{b_{11}}{s-a_{11}} & 0 \\ 0 & \frac{b_{22}}{s-a_{22}} \\ 0 & 0 \end{pmatrix}=\begin{pmatrix}b_{11} & 0 \\ 0 & b_{22} \\ 0 & 0  \end{pmatrix}.
\end{align*}
\end{example}

%Generally, there exist many realisations consistent with $[{Q},{P}]$. In sections~\ref{se:pf} and~\ref{sec:algqp}, we focus on finding a $[{Q},{P}]$ minimal realisation $\left(A,B,\begin{bmatrix} I & 0 \end{bmatrix}\right)$, i.e. a realisation consistent with $[{Q},{P}]$ of minimal order (that is, a realisation with minimal dimension for $A$), and hence with the lowest possible complexity.

\subsection{Realisation problem}

In general, ${Q(s)}$ and ${P(s)}$ carry more information than ${G^o(s)}$. This can be seen from the equality ${G^o}(s) = ({I}-{Q}(s))^{-1}{P}(s)$.  However, ${Q(s)}$ and ${P(s)}$ carry less information than the state-space model \eqref{eq:LTI} (see~\cite{08net_rec,robust} and Figure~\ref{fig:problems}). This leads to the problem of realisation of $[{Q},{P}]$, similar to the problem of realisation of $G^o$. Basically, just like the fact that there are infinite state-space realisations that give the same transfer function (realisation problem (1) and set red in Figure~\ref{fig:problems}), there are an infinite state-space realisations that give the same $[{Q},{P}]$ (realisation problem (2) and set magenta in Figure~\ref{fig:problems}).

\begin{figure}[h]
\centering
  \includegraphics*[width=.48\textwidth]{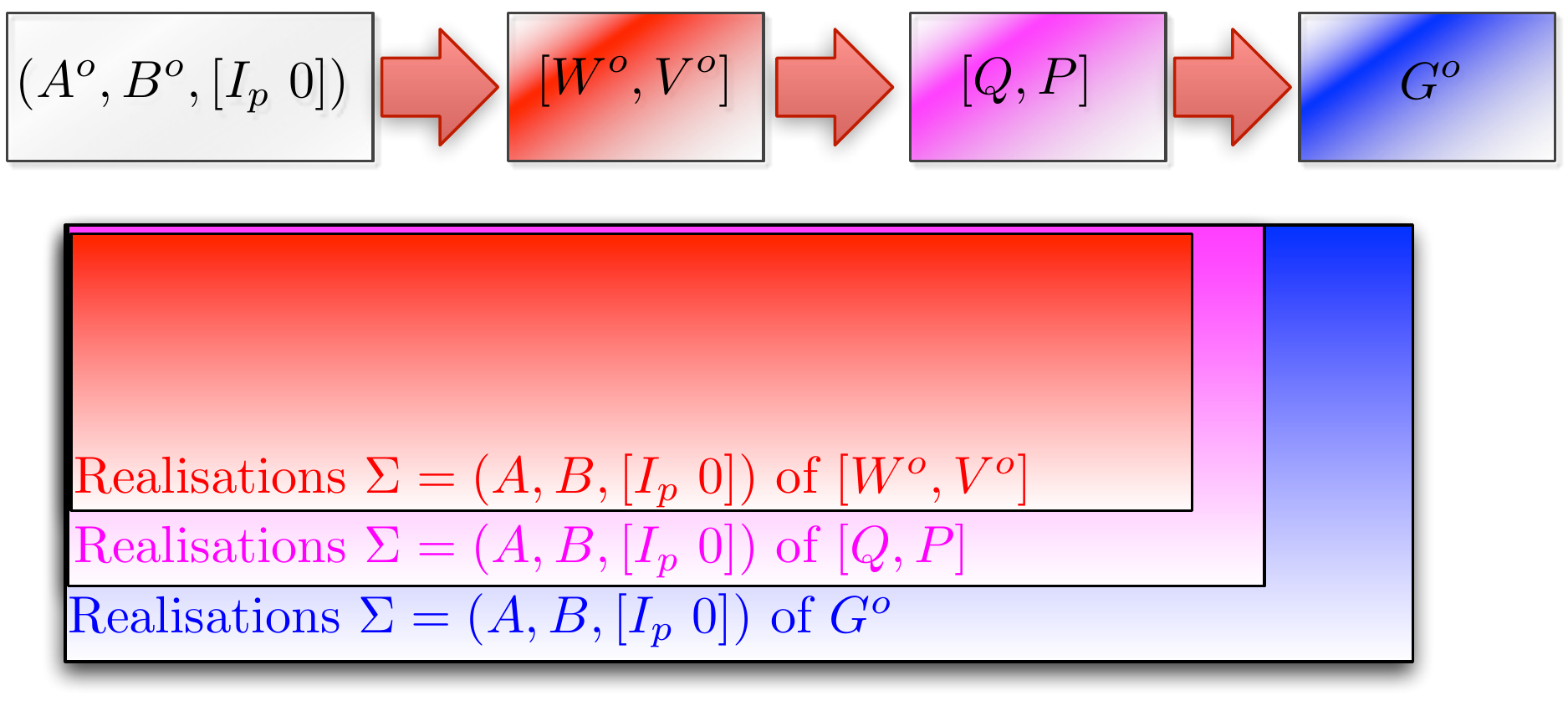}
\caption{Relations among transfer functions, dynamical structure functions and state-space realisations.} \label{fig:problems}     
\end{figure}

%\yy{a) maybe delete def, let $\Sigma$ be a realisation of $[Q,P]$, b) remove all the consistent in the paper}
%
\begin{definition}A system $\Sigma\triangleq(A,~B,~C=[I_p,0])$ is a realisation of $[{Q},{P}]$ if that $\Sigma$ gives $[{Q},{P}]$ from eq.~\eqref{eq:Q} and eq.~\eqref{eq:P}.
\end{definition}

%\begin{remark}
%not the only realisation}

%\begin{definition}
%  Consider a system characterized by a transfer function ${G}$.
%  The dynamical structure of the system can be \textit{reconstructed},
%  if there is only one admissible dynamical structure function,
%  $[{Q},{P}]$, that is consistent with ${G}$.  A
%  realisation of the dynamical structure function is defined as
%  \textit{reconstruction}. Likewise, the Boolean structure of the
%  system can be reconstructed if all admissible dynamical structure
%  functions that are consistent with ${G}$ have the same Boolean
%  structure.
%\end{definition}
%
%Given only a transfer function ${G}$, \cite{08net_rec} shows that
%dynamical structure reconstruction is not possible.  More information
%is required, i.e., dynamical structure reconstruction is possible from
%${G}$ if and only if in addition $p-1$ elements in each column of
%$[{Q}\ \ {P}]^T$ are known that uniquely specify the
%component of $[{Q},{P}]$ in the nullspace of $[{G}^T \
%\ {I}]$ (see \cite{08net_rec} for more details).

\begin{definition}
We say that a realisation $\Sigma\triangleq(A,~B,~C=[I_p,0])$ of $G$ is ${G}$-minimal if this realisation corresponds to a minimal realisation of ${G}$. We say that a realisation $\Sigma$ of $[Q,P]$ is $[{Q},{P}]$-minimal if this realisation of $[{Q},{P}]$ has the smallest order. 
%The size of $A$ is called the order of a minimal structural realisation of $[Q,P]$.
\end{definition}
%\yy{change [Q~P] to [Q,P]}

\subsection{Observability and controllability properties}
Let a system $\Sigma$ have the following form
\begin{equation}\label{eq:sigma}
\Sigma=\left(A=\begin{bmatrix}{A}_{11}&{A}_{12}\\{A}_{21}&{A}_{22}\end{bmatrix},B=\begin{bmatrix}   B_1 \\B_2
\end{bmatrix},\begin{bmatrix}
I_p & 0 \end{bmatrix}\right)
\end{equation}
be a realisation of $[Q,P]$. In this subsection, we shall introduce properties of minimal realisations of $[Q,P]$ and all the proofs in this subsection can be found in Appendix~\ref{sec:appendixA}.
.

%\yy{in the following, we can define $\Sigma$ with partition here and put only $\Sigma$ below}

\begin{proposition}\label{lemma:tra}
Let $\Sigma_1$ be a realisation of $[Q,P]$ (eq.~\eqref{eq:sigma}) and consider a linear transformation mapping $\Sigma_1=(A,B,\begin{bmatrix} I & 0 \end{bmatrix})$ to $\Sigma_2=(T^{-1}AT,T^{-1}B,\begin{bmatrix} I & 0 \end{bmatrix}T)$,  $\Sigma_2$ is also a realisation of $[Q,P]$ for any $T$ with the following form \begin{equation}\label{eq:t}
T=\begin{bmatrix}  I & 0 \\
0 & T_2
\end{bmatrix},
\end{equation}
for any invertible matrix $T_2$.
\end{proposition} 

\begin{remark}
According to the above proposition, one can apply linear transformations to the hidden states without changing the dynamical structure function. \end{remark}
Similar to minimal realisation of transfer functions, based on Proposition~\ref{lemma:tra} we can define the following hidden observability and controllability concepts. 

\begin{definition}[Hidden Observability]
Given a realisation  $\Sigma$ of $[Q,P]$, we say it is hidden observable if and only if $[A_{22},~A_{12}]$ is observable.
\end{definition}

\begin{definition}[Hidden Controllability]
Given a realisation $\Sigma$ of $[Q,P]$, we say it is hidden controllable if and only if $\left[A_{22},~[A_{21},~B_2]\right]$ is controllable.  
\end{definition}

From these two definitions, we can show that if a realisation $\Sigma$ is $[Q,P]$-minimal then it is both hidden observable and controllable.

\begin{remark}
Linear transformations of the form $T$ in eq.~\eqref{eq:t} do not change the hidden observability and hidden controllability of a system.
\end{remark}

\begin{proposition}\label{th:hidden}
If a realisation $\Sigma$ of $[Q,P]$ is minimal, then it is hidden observable and hidden controllable.
\end{proposition}
\begin{proof}
It is easy to show by contradiction. $\hfill \square$
\end{proof}

\begin{remark}
Note that a realisation $\Sigma$ of $[Q,P]$ can be hidden observable and hidden controllable and not necessarily $[Q,P]-$ minimal. 
\end{remark}

\begin{proposition}\label{th:obserable} 
If a system $\Sigma$ is hidden observable, then it is observable.
\end{proposition}

Based on the above proposition, we can show the following Corollary.
\begin{corollary}
Given a minimal realisation $\Sigma$ of $[Q,P]$, then the order of this realisation is equal to the order of $G=(I-Q)^{-1}P$ if and only if $\Sigma$ is controllable.
\end{corollary}
%\yy{need to double check everything}

%\begin{remark}
%When a minimal realisation of a given transfer function has lower dimension than that of a given dynamical structure functions that is consistent to the transfer function in question.
%\end{remark}

\section{PROBLEM FORMULATION}\label{se:pf}
From here on, the paper assumes that the dynamical structure functions $[Q,P]$ and the transfer function $G$ are known, and  the original state-space realisation~\eqref{eq:LTI} is unknown. We then proceed to search for a minimal realisation of $[Q,P]$. Just like the minimal realisation of a transfer function, the underlying principle to find a $[{Q},{P}]$-minimal realisation is to search for a realisation with the minimal number of hidden states.  The rest of the paper aims to solve the following problem.

\begin{problem}\label{prob:main}
{\em [}{\bf Minimal $[Q,P]$ realisation}{\em ]} Given a dynamical structure function $[Q,P]$, find a minimal realisation $\Sigma\triangleq(A,~B,~C=[I_p~0])$ of $[Q,P]$. 
\end{problem}

\begin{remark}
From the above definitions, the order of a minimal structural realisation of $[Q,P]$ is always higher or equal to that of a minimal realisation of a transfer function $G=(I-Q)^{-1}P$. 
\end{remark}

%\yy{change}
%After exploring properties of dynamical structure functions, we now come back to Problem~\ref{prob:main}. 
Note that the original transfer matrices $[{W}^o, {V}^o]$ cannot be reconstructed from the  dynamical structure function $[{Q},{P}]$ since there is no information regarding the diagonal proper transfer function matrix ${R}^o$. Hence, choosing an arbitrary diagonal proper transfer function matrix ${R}$ leads to an arbitrary $[W,~V]$ from the following equation 
\begin{equation}\label{eq:Rrealization}
[{W},{V}]= [(s{I} - {R}){Q} +{R},(s{I}-{R}){P}],
\end{equation} 
which is obtained from reversing the steps in equations~(\ref{eq:Q}) and~(\ref{eq:P}). Note that, in general, $[W,~V]$ will be different from $[{W}^o, {V}^o]$.
A realisation of $[W,~V]$ is given by
\begin{equation}\label{eq::WVrealization}
  [{W},{V}]=[{A}_{11},{B}_1]+{A}_{12}(s{I}-{A}_{22})^{-1}[{A}_{21},{B}_2]
\end{equation}
where ${A}$ and ${B}$ are state-space matrices, structured similarly to equation~(\ref{eq:LTI}).
Again, this realisation is, in general, different from~(\ref{eq:LTI}), since it is not possible to recover~(\ref{eq:LTI}) from $[{Q},{P}]$ alone.

%\yy{need to add a proposition which shows this covers all cases.}

\begin{remark}
Any realisation of $[Q,P]$ can be obtained from eq.~\eqref{eq:Rrealization} and eq.~\eqref{eq::WVrealization}.
\end{remark}

The idea for solving Problem~\ref{prob:main} is to use a state-space realisation approach to find an ${R}^*$ that  minimises the order of $[{W},{V}]$. Such realisation is also a $[{Q},{P}]$ minimal realisation. Mathematically, the problem can be reformulated according to finding such ${R}^*$ 
$$R^*=\text{argmin}_{R\in\mathcal{D}_p} \text{deg} [W,V],$$
where deg is the McMillan degree \cite{zdg} and $\mathcal{D}_p$ is the set of all proper diagonal transfer matrices with dimension $p$ (the number of measured states) that admits a diagonal realisation. This is equivalent to finding $R^*$ from the following equation
\begin{equation}\label{eq:D}
R^*=\text{argmin}_{R\in\mathcal{D}_p} ~\text{deg}\left\{(s{I}-{R})s^{-1}[s{Q},s{P}]+[{R},{0}]\right\}.
\end{equation}
This non-convex optimisation is, in general, hard to solve directly. 
Note that a random choice of a proper diagonal transfer function matrix ${R}$ is likely to result in additional zeros in $[sI-W,V]$.
\begin{proposition}\label{prop:control}
If $[sI-W,V]$ has a zero, then for any realisation $(A,B,C=[I_p~0])$ obtained from eq.~\eqref{eq::WVrealization}, $[A,B]$ is not controllable.
\end{proposition}
\begin{proof}
See Appendix~\ref{sec:appendixA}. $\hfill \square$ 
\end{proof}
\begin{remark}
Proposition~\ref{prop:control} shows that these additional zeros in $[sI-W,V]$ lead to unnecessary uncontrollable modes in the realisation $(A,B,C=[I_p~0])$ which means that 
$(A,B,C=[I_p~0])$ is not a minimal realisation of $[Q,P]$.
\end{remark}

\begin{remark}
There might be many choices for ${R}^*$ that minimise
the order of minimal realisations of $[{W}~{V}]$.
%a chosen
%${R}^*$ may be different from ${R}^o$. 
\end{remark}

\begin{remark}
After solving the problem in eq.~(\ref{eq:D}) and obtaining a minimal realisation of $[Q,P]$, we can use Proposition~\ref{lemma:tra} to find other minimal realisations of $[Q,P]$. 
\end{remark}
 
%From the test using Keith's code, we found that for random choice of $(A,~B,~C,~D)$, the probability of $[I-Q,P]$ has zero is high! This indicates that due to the special construction of $[I-Q,P]$, the assumption in Proposition~\ref{prop:d} is actually strong. We shall now consider the scenario when $[I-Q,P]$ has zeros.

Next, we shall convert the optimisation in eq.~\eqref{eq:D} into a simpler form that explores the structure of the optimisation. To start,  
let 
\begin{equation}\label{eq:ep}
\mathcal{E}_p\triangleq\{N| N=(I-R/s), \forall R\in\mathcal{D}_p\}
\end{equation}
and note that there is an one-to-one map between $\mathcal{E}_p$ and $\mathcal{D}_p$.

\begin{proposition}\label{prop:convert}
For any $[I-Q,P]$ with full normal row rank, the following equality holds
\begin{align}
&\min_{R\in\mathcal{D}_p} ~\text{deg}\left\{(s{I}-{R})s^{-1}[s{Q}~s{P}]+[{R}~{0}]\right\}\nonumber\\
&=\min_{N\in\mathcal{E}_p} ~\text{deg} \left\{N[I-{Q},~{P}]\right\}-p.
\end{align}
\end{proposition}
\begin{proof}
See Appendix~\ref{se:AppB}. 
\end{proof}

From the above proposition, the next section shall focus on solving
\begin{equation}\label{eq:Np}
N^* \triangleq\text{argmin}_{N\in\mathcal{E}_p} \{\text{deg}N[I-Q,P]\}.
\end{equation}

%
%\begin{assumption}\label{ass:2}
%$[I-Q,P]$ does not have zeros at $0$.
%\end{assumption}
%This assumption will be relaxed, we assume it for now to better illustrate the idea. 
%\begin{remark}
%The objective function in eq.~\eqref{eq:D} can be equivalently solved from eq. ~\eqref{eq:N}. Once $N^*$ has been solved, $R^*$ can be computed correspondingly. %If $[Q,P]$ do have poles at $0$, Proposition~\ref{prop: R} still follows by slightly modifying $$R^*=\textbf{argmin}_{R\in\mathcal{D}_p} ~\text{deg} \left\{(s{I}-{R})(s-a)^{-1}[I-{Q}~{P}]\right\}$$
%%where $a\in\mathbb{R}$ is not a pole of $[I-{Q}~{P}]$.
%\end{remark}

\section{MAIN ALGORITHM FOR OBTAINING A MINIMAL REALISATION OF $[Q,P]$}
\subsection{Analysis}\label{se:4.1}
This section proposes an algorithm to solve the optimisation in eq.~\eqref{eq:Np}. It follows that
\begin{align}
&\text{deg}\left\{N[I-{Q},{P}]\right\}=\text{deg}\left\{N\right\}+\text{deg}\left\{[I-{Q},{P}]\right\}\nonumber\\
&-\text{$\#$ of cancelled zeros of $[I-Q,P]$ by cascading} \nonumber\\
 &- \text{$\#$ of cancelled poles of $[I-Q,P]$ by cascading}. \label{eq:degree}
\end{align}
Next, we shall derive conditions on $N(s)$ for cancelling zeros and poles of $[I-Q,P]$. 
\begin{assumption}\label{ass:1}
Assume that $[I-Q,P]$ only has simple poles and does not have the same poles and zeros. 
\end{assumption}
%\yy{Can we say that this makes the derivations cleaner but is not essential? Basically, does the general case follow as well, but is a lot more messy? If yes, we should say it. If not, we should make a small comment that this is not that restrictive.}

Since $Q,~P$ are strictly proper, a minimal realisation of $[I-{Q},{P}]$ has the following form: ${C}_1({A}_1-sI)^{-1}{B}_1+D_1$. When $[I-Q,P]$ has $l$ simple poles, Gilbert's realisation \cite{gilbert} gives
\begin{equation}\label{eq:gilbertqp}
[I-{Q},{P}]=\sum_{i=1}^l
\frac{{K}_i}{s-\lambda_i}+\lim_{s\rightarrow\infty}[I-{Q},{P}],
\end{equation}
where ${K}_i =
\lim_{s\rightarrow\lambda_i}(s-\lambda_i)[I-{Q},{P}]$ and has
rank $1$, since we are assuming that $[I-{Q},{P}]$ has simple poles. Consider the following matrix decomposition for ${K}_i$:
\begin{equation}\label{eq:E}
{K}_i={E}_i{F}_i,~\forall i,
\end{equation}
where ${E}_i \in \mathbb{C}^{p}$ and
${F}_i=({E}_i^T{E}_i)^{-1}{E}_i^T{K}_i$. Then ${A}_1=\text{diag}\{\lambda_i\}\in\mathbb{C}^{l\times l}$, ${B}_1=\begin{bmatrix} {F}^T_1 & {F}^T_2 & \ldots &
  {F}^T_l \end{bmatrix}^T$, ${C}_1=\begin{bmatrix}
  {E}_1 & {E}_2 & \ldots & {E}_l \end{bmatrix}$ and
${D}_1=\lim_{s\rightarrow\infty}[I-{Q},{P}]=[I,0]$.

Similarly, $N(s)$ is a diagonal transfer matrix with its minimal realisation $(A_2,B_2,C_2,I)$. Without loss of generality, assume the matrix ${A}_2$ is diagonal  (otherwise, a linear transform can diagonalise $A_2$ without changing $N(s)$). A minimal realisation of a diagonal transfer matrix can be obtained from a composition of Gilbert realisations of all transfer functions on the diagonal. Let $N[m,m]=c_m(sI-a_m)^{-1}b_m$ where $(a_m\in\mathbb{C}^{k_m\times k_m}, b_m \in \mathbb{C}^{k_m\times 1})$ is the minimal realisation of the the $m^{th}$ transfer function on the diagonal. Then a minimal realisation of $N$ has the form
\begin{align}
&A_2=\text{diag} [a_1,~\ldots,~a_k], ~B_2=\text{diag} [b_1,~\ldots,~b_k],\nonumber\\
&C_2=\text{diag} [c_1,~\ldots,~c_k],\label{eq:realisationofN}
\end{align}
 (with $\sum_{m=1}^k k_m=r$, where $r$ is the McMillan degree of $N(s)$). In the following results, let $\mathcal{B}(\cdot)$ be the Boolean operator which maps a matrix/vector to a Boolean one.

\begin{theorem}\label{thm:cancelz}
Under Assumption~\ref{ass:1}, if a zero  $\lambda_i$ of $[I-Q,P]$ (with direction $v_i^T$) is cancelled by cascading a system ${N}(s)={C}_2({A}_2-sI)^{-1}{B}_2+{I}$, then $N[j,j](s)$ has a pole at $\lambda_i$ for any $j$ such that $\mathcal{B}(v_i^T)[j]\neq0$.
\end{theorem}
\begin{proof}
If a zero of $[I-Q,P]$, say $\lambda_i$, is cancelled by cascading a system ${N}(s)=(s{I}-{R}^*)s^{-1}\triangleq {C}_2({A}_2-sI)^{-1}{B}_2+{I}$, then the realisation of the cascaded system $(s{I}-{R^*})s^{-1}[I-{Q},{P}]$ loses controllability. 

In this case, it follows that there exists a nonzero vector ${z}_i^T= [{z}_{1,i}^T, {z}_{2,i}^T]$ such that
$$\begin{bmatrix} {z}^T_{1,i}& {z}^T_{2,i} \end{bmatrix}\begin{bmatrix} {A}_1-\lambda_i{I} & {0}  & B_1 \\ {B}_2
  {C}_1 & {A}_2-\lambda_i{I} & B_2[I, ~0]
\end{bmatrix}=0.$$
This leads to 
\begin{itemize}
\item[1.] $${z}^T_{2,i}({A}_2-\lambda_i{I})=0,$$
which indicates that ${z}_{2,i}$ is an eigenvector of ${A}_2$ corresponding to $\lambda_i$. 
\item[2.] 
\begin{equation}\label{eq:temp}
\begin{bmatrix} {z}^T_{1,i}& {z}^T_{2,i}B_2 \end{bmatrix}\begin{bmatrix} {A}_1-\lambda_i{I} & B_1 \\ 
  {C}_1 & [I, ~0]
\end{bmatrix}=0.\end{equation}
\end{itemize}
Notice that, 
\begin{align*}
\begin{bmatrix} {A}_1-s{I} & {B}_1\\ {C}_1& [{I},~0]
\end{bmatrix}
&\begin{bmatrix} {I} & -({A}_1-s{I})^{-1}B_1 \\
0& {I}
\end{bmatrix}\\ 
&=\begin{bmatrix} {A}_1-s{I} & 0\\ C_1 &
{[I-Q(s),P(s)]}
\end{bmatrix},
\end{align*}
and, since $\lambda_i$ is not a pole of $[I-Q(s),~P(s)]$, it follows from eq.~\eqref{eq:temp} that
\begin{equation}\label{eq:reqs}
{z}^T_{2,i}B_2{[I-Q(\lambda_i),P(\lambda_i)]}=0
\end{equation}
By definition, $\lambda_i$ is a zero of ${[I-Q(s),P(s)]}$ if
there exists a $v_i^T$ such that 
\begin{equation}\label{eq:zero}
v_i^T{[I-Q(\lambda_i),P(\lambda_i)]}=0,
\end{equation}
By comparing eqs.~(\ref{eq:reqs}) and~(\ref{eq:zero}) we conclude that $v_i^T={z}^T_{2,i}B_2$ is the vector associated with the zero direction of ${[I-Q(\lambda_i),~P(\lambda_i)]}$. 
Then, it also follows that  
\begin{equation}\label{eq:vector}
\mathcal{B}({z}^T_{2,i}B_2)=\mathcal{B}(v_i^T).
\end{equation}
%On the other hand, ${z}^T_{2,i}B_2=[1~2]$ and ${z}^T_{2,m}A=\lambda_m A$ where ${z}^T_{2,i} \in \mathbb{R}^{l}$ and $B_2\in\mathbb{R}^{l\times n}$, we can scale the left eigenvector arbitrarily but its Boolean structure preserves, i.e. 
%$\mathcal{B}(z_{2,i}^TB_2)=\mathcal{B}(v^T_i)$. 
Since $a_m$ in eq.~\eqref{eq:realisationofN} are diagonal matrices for all $m$, then without loss of generality 
\begin{align*}
z^T_{2,i_m}&=\begin{bmatrix}1&0\ldots &0\end{bmatrix},~\forall m\\
z^T_{2,i}&=\begin{bmatrix}z^T_{2,i_1}&z^T_{2,i_2}\ldots &z^T_{2,i_k}\end{bmatrix}
\end{align*}
if $a_m$ has an eigenvalue as $\lambda_i$. 
Since $B_2$ also has a block diagonal structure,  we have 
\begin{align*}
{z}^T_{2,i}B_2&={z}^T_{2,i}\text{diag} [b_1,~\ldots,~b_k]\\
&=\left[b_1[:,1],~b_2[:,1], \ldots,~ b_k[:,1]\right].
\end{align*}
This implies that the $j^{th}$ nonzero elements in $v_i^T$ corresponds to a nonzero element in $b_j[:,1]$ which further implies that $\lambda_i$ is a pole of $N[j,j](s)$, the $j^{th}$ transfer function on the diagonal of $N(s)$. $\hfill \square$
\end{proof}

\begin{theorem}\label{thm:cancelp}
Under Assumption~\ref{ass:1}, if a pole $\lambda_i$ of $[I-{Q},{P}]$ is cancelled by cascading a system ${N}(s)={C}_2({A}_2-sI)^{-1}{B}_2+{I}$, then\begin{equation}\label{eq:reqne}
{N}(\lambda_i){E}_i={0},
\end{equation}
where $E_i$ is defined in eq.~\eqref{eq:E}.
\end{theorem}
\begin{proof}
If a pole of $[I-Q,P]$, say $\lambda_i$, is cancelled by ${N}(s)=(s{I}-{R}^*)s^{-1}\triangleq {C}_2({A}_2-sI)^{-1}{B}_2+{I}$, then the realisation of the
cascade $(s{I}-{R})s^{-1}[I-{Q}~{P}]$ loses
observability. In this case, it follows that there exists a nonzero vector
${w}_i= [{w}_{1,i}^T, {w}_{2,i}^T]^T$ such that
\begin{equation}\label{eq:chobv}
\begin{bmatrix} {A}_1-\lambda_i{I} & {0} \\ {B}_2
  {C}_1 & {A}_2-\lambda_i{I} \\ {C}_1& {C}_2
\end{bmatrix}\begin{bmatrix} {w}_{1,i}\\ {w}_{2,i} \end{bmatrix}={0}.
\end{equation}
The first equation in eq.~\eqref{eq:chobv} shows that ${w}_{1,i}$ is an eigenvector of
${A}_1$ corresponding to $\lambda_i$. Since
${A}_1$ is diagonal, we can directly compute ${w}^T_{1,i}=\begin{bmatrix}0&\ldots&0&1_{i^{th}}&0&\ldots&0\end{bmatrix} \in
\mathbb{R}^{1 \times l}$. Therefore we have
$$\begin{bmatrix} {A}_2-\lambda_i{I} & {B}_2\\ {C}_2& {I}
\end{bmatrix}\begin{bmatrix} {w}_{2,i} \\
  {C}_1{w}_{1,i} \end{bmatrix}={0}.$$
Noticing that ${C}_1{w}_{1,i}= {E}_{i}$ from eq.~\eqref{eq:E}, that
\begin{align*}
&\begin{bmatrix} {I} & {0} \\
-{C}_2({A}_2-s{I})^{-1}& {I}
\end{bmatrix}\begin{bmatrix} {A}_2-s{I} & {B}_2\\ {C}_2& {I}
\end{bmatrix}
=\begin{bmatrix} {A}_2-s{I} & {B}_2\\ {0} &
{N(s)}
\end{bmatrix},
\end{align*}
and that $\lambda_i\neq0$ is not a pole of $N(s)$, we obtain
${N}(\lambda_i){E}_i={0}.$$\hfill \square$
\end{proof}

Based on Theorem~\ref{thm:cancelp}, we have the following Corollary.
\begin{corollary}\label{coro:cancelp}
If a pole $\lambda_i$ of $[I-{Q},{P}]$ is cancelled by cascading a system ${N}(s)={C}_2({A}_2-sI)^{-1}{B}_2+{I}$, then $N[j,j](s)$ has a zero at $\lambda_i$ for any $j$ such that $\mathcal{B}(E_i)[j]\neq0$.
\end{corollary}

\begin{remark}
In summary, designing ${R}^*$ to cancel any pole
$\lambda_i$ of $[I-{Q},{P}]$ is equivalent to imposing that
eq.~\eqref{eq:reqne} holds. 
\end{remark}

\begin{remark}
The Boolean structure of $E_i$, $\mathcal{B}({E}_i)$ imposes constraints on the diagonal terms in $N(s)$ for cancelling the poles of $[I-Q,P]$.
\end{remark}

\subsection{Algorithm to find $N^*$}\label{se:4.2}
Following the derivations and analysis of the previous section, we shall propose an algorithm to directly answer the question in Problem~\ref{prob:main}: given $[I-{Q}~{P}]$, what is the maximal number of poles that can be cancelled by left multiplication of ${N}(s)$, bearing in mind that $N(s)\in\mathcal{E}_p$? 

Recall from eq.~\eqref{eq:degree} that 
\begin{align*}
&\text{deg}\left\{N[I-Q,P]\right\}=\text{deg}\left\{N\right\}+\text{deg}\left\{[I-Q,P]\right\}\nonumber\\
&-\text{$\#$ of cancelled zeros of $[I-Q,P]$ by cascading}\nonumber\\
 &- \text{$\#$ of cancelled poles of $[I-Q,P]$ by cascading}, 
\end{align*}
Since the $\text{deg}\left\{[I-Q,P]\right\}$ is known and fixed, let $f(N)\triangleq N[I-Q,P]$ then the optimisation problem in above equation becomes
\begin{align*} 
&\min_{N\in\mathcal{E}_p}\text{deg}f(N)=
\min_{N\in\mathcal{E}_p}\{-\text{$\#$ of cancelled zeros of $[I-Q,P]$}\nonumber\\
&- \text{$\#$ of cancelled poles of $[I-Q,P]$}+\text{deg}\left\{N\right\}\}.
\end{align*}
Note that any $R\in\mathcal{D}_p$ can be written as 
\begin{equation}
R=\text{diag}\{\frac{n_1}{d_1},\ldots,\frac{n_p}{d_p}\}, 
\end{equation}
where $n_i$ and $d_i$ are coprime factors for all $i=1,2,\ldots,p$. Then 
\begin{align}
N=I-R/s&=\text{diag}\{\frac{sd_1-n_1}{sd_1},\ldots,\frac{sd_p-n_p}{sd_p}\} \label{eq:Nnew1}\\
&\triangleq\text{diag}\{\frac{\hat{n}_1}{sd_1},\ldots,\frac{\hat{n}_p}{sd_p}\} \label{eq:Nnew2}
\end{align}
where $\text{deg}\hat{n}_i=\text{deg}d_i+1$ for all $i$. Next, we shall propose how to design $d_i$ and $\hat{n}_i$ based on Theorems~\ref{thm:cancelz} and~\ref{thm:cancelp}.

To minimise the above cost function, we are aiming to have maximal number of zeros of $[I-Q,P]$ as poles of some diagonal elements in $N(s)$ from Theorem~\ref{thm:cancelz}, since this a) will not increase the McMillan degree of $f(N)$ and b) gives more degrees of freedom in zeros of $N(s)$ to cancel the poles of $[I-Q,P]$ (and therefore minimise the McMillan degree of the cascaded system), since the number of zeros in $N[i,i](s)$ equals the number of its poles. Moreover, to cancel one pole $\lambda_i$ of $[I-Q,P]$, Corollary~\ref{coro:cancelp} requires that $N[j,j](s)$ has a zero at $\lambda_i$ if $\mathcal{B}(E_i)[j]\neq0$. Assuming the sparsity of $\mathcal{B}(E_i)$ is $k$ ($k\ge1$), this would then lead to an increase of the degree of $N(s)$ by $k$ since $N$ is a diagonal transfer matrix with every element in $\mathcal{E}_1$.

Based on the analysis above and Theorems~\ref{thm:cancelz} and~\ref{thm:cancelp}, to maximise the the number of poles in $[I-Q,P]$ that can be cancelled, one should a) design poles of $N$ to cancel all the zeros of $[I-Q,P]$ and b) use the degrees of freedom in designing zeros of $N$ to cancel as many poles of $[I-Q,P]$.

The next question is then how to solve the optimisation problem in b) once a) is done. Technically, we can use table~\ref{table:example}, which is generated as follows. Column $i$ corresponds to the $i^{th}$ place on the diagonal of the to-be-designed $N(s)$ and the rows are the poles of $[I-Q,P]$,  $p_i^0$, in any order.  The intersection of the $i^{th}$ row and the $j^{th}$ column is a Boolean value corresponding to $\mathcal{B}(E_i)$, the Boolean map of the corresponding direction of the $i^{th}$ pole. It is $1$ if we require the $j^{th}$ element of $N(s)$ to have a zero at $p_i^0$ to cancel the $i^{th}$ pole $p_i^0$ of $[I-Q,P]$. Hence, table~\ref{table:example} shows the requirements to cancel each of the poles as expressed in eq.~\eqref{eq:reqne}. 
 
\begin{table}[htbp]
 \centering
 \begin{tabular}{|c|c|c|c|c|c|}
   \hline
    Poles & Place $1$& Place $2$&\ldots& Place $p-1$&  Place $p$ \\ \hline
   $p^0_{1}$ & $1$ & $0$ & \ldots & $1$ & $0$ \\ \hline
   $p^0_{2}$ & $1$ & $0$ & \ldots& $0$ & 1 \\ \hline
   $p^0_{3}$ & $1$ &  $0$ & \ldots & $0$ &  1 \\ \hline
   $\vdots$ & $\vdots$ & $\vdots$ &$\ddots$&$\vdots$ & $\vdots$\\ \hline
   $p^0_{l-1}$ & 0& 0&\ldots &0 &1 \\ \hline
   $p^0_{l}$ &0   & 0  &\ldots&1 &1  \\ \hline
 \end{tabular}
 \caption{Table for computing the maximum number of pole-zero cancellations.}
  \label{table:example}
\end{table}

%\todo{require that [I-Q,P] has different poles/zeros, right?}

We then maximise the largest number of rows such that, for any column, the summation of the elements on the selected rows is less or equal to a constant obtained from eq.~\eqref{eq:vector}. Choosing a row is equivalent to cancelling the corresponding pole in $[I-Q,P]$. So, the question is how to cancel the largest number of poles without introducing more poles in the cascaded systems of $N(s)$ and $[I-Q,P]$?

Mathematically, let $\psi[j]$ be the maximum number of zeros allowed for the $j$th diagonal element of $N(s)$ and let $T[i,j]\triangleq \mathcal{B}(E_i)[j]\in\{0,1\}$ be the binary element in the
$i$th row and $j$th column of Table~\ref{table:example}. Then, the original problem can be written as the following optimisation
\begin{align}
\max~ &k=\text{card} \{i_1,\ldots,i_k\} \label{eq:k}\\
\text{s.t.},~&\sum_{h=1}^k  T[i_h,j] \le \psi[j],~\forall j, \nonumber \\
& \{i_1,\ldots,i_k\}\subseteq \{1,\ldots,l\},\nonumber
\end{align}
where card is the cardinality of a set.

%\begin{proposition}\label{prop:d}
%Assume $[{Q}~{P}]$ is defined above and $[I-Q,P]$ has zeros $\{p^0_1,~\ldots,~p^0_l\}$, a minimal order realisation of $[{W}~{V}]$ in~\eqref{eq::WVrealization} can be achieved using a diagonal transfer matrix $N^*$.
%\end{proposition}
%\begin{proof}
%The proof contains two parts
%\todo{add}
%\end{proof}
%
%
%

Let $x=[x_1,~x_2,\ldots,~ x_l]$ be binary numbers. We can reformulate eq.~\eqref{eq:k} to the following optimisation problem
\begin{align}
\max ~&x^T \vect{1}\\
\text{s.t.},~&x^T T \le \psi \nonumber,
\end{align}
where the inequality in the constraint is element-wise. This is in a standard form of binary integer programming. When the number of poles is small, the problem is easy to solve, as we can use the exhaustive attack method to go through all the possible cases and find the largest $k$. In general, however, it is an integer optimisation problem and can be viewed as a $n$-dimensional Knapsack problem and therefore NP-hard. We can use, for example, the standard Balas algorithm \cite{balas} to solve it. Once we have determined  $\{i_1,~\ldots,~i_k\}$, we can compute the corresponding zeros and poles of $N^*(s)$ and then solve for  $R^*(s)$.

The above analysis can be summarised with the following algorithm to find a minimal realisation of $[Q,P]$.
\begin{algorithm}[!]
%\begin{algorithmic}
\caption{Minimal $[Q,P]$ realisation}
\begin{step}
Compute the zeros $z^0_i$ of $[I-Q,P]$ and the corresponding directions $v_i^T$. Take the Boolean structure $\mathcal{B}(v_i^T)$, and define the vector  $\psi=\sum\mathcal{B}(v_i^T)+{1}^T$;
\end{step}
\begin{step}
Find a Gilbert realisation of $[I-Q,P]$ and find the conditions in eq.~\eqref{eq:reqne}
for cancelling the poles $p_i^0$;
\end{step}
\begin{step}
Build a table for the cancelling conditions from Step $2$ and compute the maximum 
number of poles that can be cancelled from eq.~\eqref{eq:k};
\end{step}
\begin{step}
Determine $N^*(s)$ based on the table and obtain $[I-X^*,Y^*]=N^*(s)[I-Q,P]$;
\end{step}
\begin{step}
Compute $[W,V]=s[X^*,Y^*]$.
\end{step}
\begin{step}
Find a minimal realisation of $[W,~V]$ and obtain corresponding $A,~B:$ 
$[{W}~{V}]=[{A}_{11}~{B}_1]+{A}_{12}(s{I}-{A}_{22})^{-1}[{A}_{21}~{B}_2].$
\end{step}
\label{alg:iqpzero}
%\end{algorithmic}
\end{algorithm}

\subsection{Special case:  $[I-Q,P]$ does not have zeros}\label{sec:special}
A special case of Algorithm~\ref{alg:iqpzero} is that when $[I-Q,P]$ does not have any zeros and simple poles. In this case, we have the following proposition. 
\begin{proposition}\label{prop:d}\cite{yetac}
Assume $[{I-Q},{P}]$ only has simple poles as in Assumption~\ref{ass:1} and does not have any zeros. A minimal realisation of $[Q,P]$
%\begin{equation}\label{eq::WVrealisation}
%  [{W}~{V}]=[{A}_{11}~{B}_1]+{A}_{12}(s{I}-{A}_{22})^{-1}[{A}_{21}~{B}_2]
%\end{equation}
can be obtained using a constant diagonal matrix ${R}^*$ in eq.~\eqref{eq:Rrealization} and in eq.~\eqref{eq::WVrealization}.\end{proposition}

Basically, when $[I-Q,P]$ does not have zeros and simple poles, $R^*$ is a constant matrix. Hence, there is a much simpler algorithm to obtain the maximum number of cancelling poles, rather than solving the linear integer programming. The problem reduces to the following
\begin{align}
\max~ &k=\text{card} \{i_1,\ldots,i_k\} \label{eq:knonzero}\\
\text{s.t.},~&\sum_{h=1}^k  T_{i_h,j} \le 1,~\forall j, \nonumber \\
& \{i_1,\ldots,i_k\}\subseteq \{1,\ldots,l\}.\nonumber
\end{align}
This problem still takes exponential-time to solve. There exist, however, a number of graph theoretical tools to solve it efficiently. As explained in \cite{GA}, an undirected graph is
denoted by $\mathcal{G}=(\mathcal{V},\mathcal{E})$
where $\mathcal{V}=\left\{\nu_{1},\ldots,\nu_{l}\right\}$ is the set
of nodes and $\mathcal{E}\subset\mathcal{V}\times\mathcal{V}$ is the set of edges. For our purposes, we construct an undirected graph $\mathcal{G}_a$ using the following rules:
\begin{itemize}
\item A node is associated with each vector in the set $\{{p}^0_1,\cdots,{p}^0_l\}$. There are thus $l$ nodes in the considered graph. \item An undirected edge $(i,j)$ is drawn between node $i$ and node $j$ if the equality $\mathcal{B}({E}_i)^T\mathcal{B}({E}_j)=0$ is satisfied.
\end{itemize}
Then, the maximum cardinality of $\{i_1,~i_2,\ldots,i_k\}$ in eq.~\eqref{eq:knonzero} corresponds to the maximum number of nodes in a complete
subgraph $K_n$ of the graph $\mathcal{G}_a$. Although the problem of finding the largest complete subgraphs in an undirected graph is an NP-hard problem, solutions have been proposed in \cite{link}. For an arbitrary graph, the fastest algorithm has a complexity of $\mathcal{O}(2^{n/4})$~\cite{rob}. Hence, we can use these methods to obtain one of the largest complete subgraphs and consequently compute the corresponding set $\{i_1,~i_2,\ldots,i_k\}$ with cardinality $k$.

%In this case, we have have the following algorithm
%\begin{algorithm}[!]
%%\begin{algorithmic}
% \caption{Minimal $[Q,P]$ realisation, when $[I-Q,P]$ does not have zeros  {\color{red} and simple poles?}}
%\begin{steps}
%Compute the zeros $z^0_i$ of $[I-Q,P]$  {\color{red} I thought there were none by assumption??}and the corresponding directions $v_i^T$. Take the Boolean structure $\mathcal{B}(v_i^T)$  {\color{red} and use this where?}, and define a vector that $s={1}^T$;
%\end{steps}
%\begin{steps}
%Find a Gilbert realisation of $[I-Q,P]$ and find the conditions in eq.~\eqref{eq:reqne}
%for cancelling the poles $p_i^0$;
%\end{steps}
%\begin{steps}
%Determine $N^*(s)$ based on the table and obtain $[W,~V]=N^*(s)[I-Q,P]$;
%\end{steps}
%\begin{steps}
%Find a minimal realisation of obtained $[W,~V]$ and obtain corresponding $A,~B$ 
%$[{W}~{V}]=[{A}_{11}~{B}_1]+{A}_{12}(s{I}-{A}_{22})^{-1}[{A}_{21}~{B}_2].$
%\end{steps}
%\label{alg:iqpnzero}
%%\end{algorithmic}
% {\color{red} This algorithm looks the same as the previous. Can you please have a look? It feels you copy/pasted and forgot the adjust it...}
%\end{algorithm}
%

\section{SIMULATION}
In this section, we will illustrate the above algorithms with examples.
\begin{example}
%Assume that we have the following linear system
%\begin{equation*} \begin{array}{rcl}
%\dot{x} &=& 
%\begin{bmatrix}-3 & 0 & -1 & 0 & 0 \\
%			0 & -1 & 0 & -1 & 0 \\
%			0 & -1 & -2 & 0 & -1 \\
%		-1 & 0 & 0 & -1 & -1 \\
%			0 & -1 & 0 & 0 & -1 \end{bmatrix} x + 
%\begin{bmatrix}1 & 0 \\ 0 & 1 \\ 0 & 0 \\ 0 & 0 \\ 0 & 0 \end{bmatrix} u \\
%y &=& \begin{bmatrix} I_3 & 0 \end{bmatrix} x
%\end{array} \end{equation*}
%where $I_3$ is a $3\times3$ identity matrix. Following the above definitions of $[W^o,~V^o]$ we have 
%\begin{align*}
%[W^o ~V^o] &=  
%\begin{bmatrix} -3           &    0  &                 0 &    1 &   0\\
%                \frac{1}{s+1} &  -1+\frac{1}{(s+1)^2} &   0 &   0 &   1\\
%                  0    &       \frac{1}{s+1} &           -2 &   0 &  0 
%                  \end{bmatrix}\\
%R^o &=\text{diag}\{-3, -1+\frac{1}{(s+1)^2}, -2\}
%\end{align*}
Consider $Q, P$ with the following form
\begin{align}
Q &= \begin{pmatrix} 0 & 0 & \frac{-1}{s+3} \\ \frac{s+1}{(s+1)^3+1} & 0 & 0 \\ 0 & \frac{1}{(s+4)(s+2)} & 0 \end{pmatrix},\\
P &= \begin{pmatrix} \frac{1}{s+3} & 0 \\ 0 & \frac{(s+1)^2}{(s+1)^3+1} \\ 0 & 0 \end{pmatrix}.
\end{align}
Here is an illustration of Algorithm \ref{alg:iqpzero}.
\begin{Steps}
Compute the zeros and corresponding zero directions of $[I-Q,P]$. In this case, it only has one zero at $-1$ with a corresponding zero direction of $[0,~1,~0]$. From eq.~\eqref{eq:vector} and definition of $\psi$ we can see that $\psi =[1,~2,~1]$.
\end{Steps}
\begin{Steps}
Obtain a Gilbert realisation of $[I-Q,~ P]$ 
%\yy{if you need to save space, you can write the A matrix as A=diag...}
\begin{align*}
&A_1 = \text{diag}\{-4,-2,-2,-.5+.866i,-.5-.866i,-3\}\\
%\begin{bmatrix}           
%-4  &      0    &              0           & 0 & 0 & 0 \\
%        0         &   -2  &      0           &0 & 0 & 0 \\
%        0              &     0           &      -2 & 0 & 0 & 0\\
%        0            &       0        &           0           &-.5+0.866i & 0 & 0 \\
%        0 & 0 & 0 & 0 & -.5-0.866i & 0 \\
%        0 & 0 & 0 & 0 & 0 & -3
%                    \end{bmatrix}, \\
&B_1 = \begin{bmatrix}
             0          & 1.41 & 0 & 0 & 0 \\
        0          &   2.24 & 0 &  0  & 0 \\
 0.408 & 0 & 0 & 0 & 0.408\\
    0.488-0.423i        &         0       &          0  &      0&      -0.61-0.211i\\
        0.488+0.423i        &         0       &          0  &      0&      -0.61+0.211i\\
0 & 0 & 0 & 1 & 0
 \end{bmatrix},\\
 &C_1 = \begin{bmatrix}
    0 &   0   &        0                 &  0 & 0 & 1\\          
        0    &                0     &               0 .816 & -0.488+0.169i &     -0.488-0.169i       &          0\\
          -0.707     &        0.447        &         0 & 0 &  0  & 0
 \end{bmatrix},\\
  &D_1 = \begin{bmatrix}
  1 &  0 &  0  & 0 & 0 \\
0  & 1 & 0  & 0  & 0 \\
0 & 0 &  1& 0 & 0
\end{bmatrix}.
\end{align*}
Based on the above analysis, we can draw Table~\ref{table:example1}. There, we see that to cancel pole $p^0_1$, we need have 
a zero on the first diagonal element of $N(s)$, similarly for other poles.

\begin{table}[!]
 \centering
 \begin{tabular}{|c|c|c|c|c|c|}
   \hline
    Poles & Place $1$& Place $2$ & Place $3$\\ \hline
   $p^0_{1}= -3$ & $1$ & $0$ & $0$\\ \hline
   $p^0_{2}= -2$ & $0$ & $1$ &  $0$ \\ \hline
   $p^0_{3}=-.5+0.866i$ &$0$   & $1$  &  $0$\\ \hline
   $p^0_{4}=-.5-0.866i $ & $0$& $1$& $0$\\\hline
   $p^0_{5}=-4$ & $0$ &  $0$  & $1$ \\ \hline
   $p^0_{6}=-2$ &$0$  &$0$  & $1$ \\\hline
  \end{tabular}
\caption{Table for computing maximum cancelled poles.}\label{table:example1}
\end{table}
\end{Steps}

\begin{Steps}
Solve the following optimisation problem
\begin{align*}
\max~ &k \\
\text{s.t.},~&\sum_{h=1}^k T_{i_h,j}  \le 1,~\forall j=1,3.  \\
 ~&\sum_{h=1}^k T_{i_h,j}  \le 2,~\forall j=2.  \\
& \{i_1,\ldots,i_k\}\subseteq \{1,\ldots,l\}.
\end{align*}
where , $T_{i,j}\in\{0,1\}$ is the binary element in the $i^{th}$ row and $j^{th}$ column of Table~\ref{table:example1}. By solving the above optimisation, the optimal solution is $k=4$.  Hence, the dimension of $A$ is $p+l-k=3+6-4=5$. 
\end{Steps}

\begin{Steps}
There are several optimal solutions. Choose, for example, the solution $\{i_1,\ldots,i_k\}=\{1,3,4,5\}$. Then 
$$N^*(s)=\text{diag} \left[k_1\frac{s+3}{s}, ~ k_2\frac{s^2+s+1}{s^2+s},~   k_3\frac{s+4}{s} \right]$$ 
where $k_i$ are nonzero parameters.
\end{Steps}
\begin{Steps}
If $k_i=1$, then 
\begin{align*}
[I-W/s,~ V/s]&=N^*[I-Q,P]\\
&=\begin{bmatrix}
\frac{s+3}{s} & 0 & 0 & \frac{1}{s} & 0 \\
\frac{-1}{s(s+2)} & \frac{(s+1)^2}{s(s+2)} & 0 & 0 & \frac{s+1}{s^2+2s}\\
0 & \frac{-1}{s^2+2s} & \frac{s+4}{s} & 0 & 0 
\end{bmatrix},
\end{align*}
which gives
\begin{equation*}
[W,~ V]=\begin{bmatrix}
-3 & 0 & 0 & 1 & 0 \\
\frac{1}{s+2} & \frac{-1}{s+2} & 0 & 0 & \frac{s+1}{s+2}\\
0 & \frac{-1}{s+2} & -4 & 0 & 0 
\end{bmatrix}.
\end{equation*} 
with
$$R^*=sI-sN^*=\text{diag}\left[-3,~ \frac{1}{s+1}, ~ -4\right] 
.$$
\end{Steps}

\begin{Steps}
Find a minimal realisation of $[W,~V]$ and obtain the corresponding $A,~B$ matrices:
\begin{align*}
[{A}_{11}~{B}_1]&= \begin{bmatrix}
-3 & 0 & 0 & 1 & 0 \\
0 & 0 & 0 & 0 & 1\\
0 & 0 & -4 & 0 & 0 
\end{bmatrix},\\
{A}_{12}&= \begin{bmatrix}
0 & 0\\ 1& 0\\ 0 &1
\end{bmatrix}, A_{22}=\begin{bmatrix}
-2 & 0\\ 0& -2
\end{bmatrix},\\
[{A}_{21}~{B}_2]&=\begin{bmatrix}
1 & -1 & 0 & 0 & -1\\
0 & -1 & -1 & 0 & 0 
\end{bmatrix}.
\end{align*}
Hence, a minimal realisation has the following form:
\begin{align*}
A=\begin{bmatrix}
-3 & 0 & 0 & 0 & 0 \\
0 & 0 & 0 & 1 & 0\\
0 & 0 & -4 & 0 & 1\\
 1 & -1 & 0 & -2 & 0\\
0 & -1 & -1 & 0 & -2
\end{bmatrix},~~
B= \begin{bmatrix}
 1 & 0 \\
 0 & 1\\
 0 & 0 \\
 0 & -1\\
 0 & 0 
\end{bmatrix}.
\end{align*}
\end{Steps}

Note that, as mentioned in Step 3, there are several solutions to $N^*$. For example, chosing of the solution $\{i_1,\ldots,i_k\}=\{1,3,4,6\}$ would have lead to a different $N^*$.
However, ultimately all optimal solutions have $A,~B$ matrices of the same dimension.

\end{example}
\section{CONCLUSION}
This paper presented an algorithm for obtaining a minimal order realisation of a given dynamical structure function. This provided a way to estimate the complexity of systems by determining the minimal number of hidden states in networks. This can help understand the minimal number of unknown states interacting in a particular network.

\appendix
\section{Appendix: proof of Proposition~\ref{th:qp}, \ref{lemma:tra}, \ref{th:obserable} and~\ref{prop:control}}\label{sec:appendixA}
\begin{proof}\textbf{[Proposition~\ref{th:qp}]}
Eq.~\eqref{eq:Ds} is directly obtained from the definition of ${R}^o(s)$:
\begin{align*}
  \lim_{s\rightarrow\infty}{R}^o(s)&=\lim_{s\rightarrow\infty}\text{diag}\{{W}^o(s)\}\\
  &=\text{diag}\{\lim_{s\rightarrow\infty}{W}^o(s)\} =
  \text{diag}\{{A}_{11}^o\}
\end{align*}
Since the proofs of eq.~\eqref{eq:Qs} and~\eqref{eq:Ps} are very
similar, we focus on eq.~\eqref{eq:Qs} only. In the following, we use the fact that
for any square matrix ${M}$, if $M^n \rightarrow 0$ when $n\rightarrow +\infty$, then  
$({I}-{M})^{-1}=\sum_{i=0}^{\infty}{M}^i$. From the definition of ${Q}$ in \eqref{eq:Q},
${Q}(s)=\sum_{i=1}^{\infty}
s^{-i}{R}^{o~i-1}(s)\left({W}^o(s)-{R}^o(s)\right)$ and
${W}^o(s)={A}^o_{11}+\sum_{i=1}^{\infty}s^{-i}{A}^o_{12}{A}_{22}^{o~i-1}{A}^o_{21}$, when $s\rightarrow+\infty$.
Hence,
${Q}(s)=({A}^o_{11}-{R}^o(s))s^{-1}+{r}(s)$, in
which ${r}(s)$ is a matrix polynomial of $s$, whose largest degree is
$-2$. Finally, multiplying by $s$ on both sides and taking the limit
as $s$ goes to $\infty$ results in eq.~\eqref{eq:Qs}.  A similar
argument can be used to prove eq.~\eqref{eq:Ps}.$\hfill \square$
\end{proof}

\begin{proof}\textbf{[Proposition~\ref{lemma:tra}]}
Partition $A$ and $B$ according to the following form
$$A=\begin{bmatrix}{A}_{11}&{A}_{12}\\{A}_{21}&{A}_{22}\end{bmatrix},~ B=\begin{bmatrix}   B_1 \\B_2
\end{bmatrix},$$
From this partition, we have that 
$$T^{-1}AT=\begin{bmatrix}{A}_{11}&{A}_{12}T_2\\T_2^{-1}{A}_{21}&T_2^{-1}{A}_{22}T_2\end{bmatrix} ,T^{-1}B=\begin{bmatrix}   B_1 \\T_2^{-1}B_2
\end{bmatrix}.$$
We can then directly compute $[W,~V]$ for the transformed system and verify that such transformation will preserve $[Q,P]$ (the details of the rest of the proof are omitted).$\hfill \square$
\end{proof}

\begin{proof}
\textbf{[Proposition~\ref{th:obserable}]}
From the Popov-Belevitch-Hautus (PBH) rank test~\cite{zdg}, a matrix pair
  $({A}\in \mathbb{R}^{l \times l},{C})$ is observable iff
  \begin{equation}\label{eq:pbh}
    \begin{bmatrix} s{I}-{A}\\
      {C}\end{bmatrix}\text{has full column rank.}
  \end{equation}
  for all $s\in\mathbb{C}$. 
  
   If a realisation is hidden observable, then it implies that the
  pair $({A}_{22},{A}_{12})$ is observable from its definition, i.e.,
  \begin{equation*}  \begin{bmatrix}
      s{I}_{l-p}-{A}_{22}\\
      {A}_{12} \end{bmatrix} \text{has full column rank, $l-p$},~\forall s.
  \end{equation*}
  Hence
  \begin{equation*} \begin{bmatrix}
      s{I}-{A}_{11}&-{A}_{12}\\ -{A}_{21} &
      s{I}_{l-p}-{A}_{22}\\ {I}_p & {0}_{p \times
        (l-p)} \end{bmatrix}\text{has full column rank $l$,}~\forall s,
  \end{equation*}
  which concludes the proof. $\hfill \square$
  \end{proof}

\begin{proof}
\textbf{[Proposition~\ref{prop:control}]}
If $z_0$ is a zero of $[sI-W,~V]$ with direction $v_0^T$, by definition of zeros of a transfer function \cite{zdg},
$$v_0^T[sI-W,V]\big|_{s=z_0}=0.$$
Recall the definition of $[W,V]$
\begin{align*}
\left.v_0^{T}\left(sI-{A}_{11} - {A}_{12}\left ( s{I} -
  {A}_{22} \right )^{-1} {A}_{21}\right)\right\vert_{s=z_0}&=0\\
\left . v_0^T\left({B}_{1} +{A}_{12}\left ( s{I} - {A}_{22} \right
)^{-1} {B}_{2}\right)\right\vert_{s=z_0}&=0.
\end{align*}
Let $v_1^{T}\triangleq v_0^{T}{A}_{12}\left ( s{I} -
  {A}_{22} \right )^{-1}|_{s=z_0}$, then we have 
\begin{align}
v_0^{T}(z_0I-{A}_{11})-v_1^T{A}_{21}&=0\\
v_0^TA_{12}-v_1^T(z_0I-A_{22}) &=0 \label{eq:unob}\\
v_0^T{B}_{1} +v_1^T {B}_{2}&=0,
\end{align}
where eq. ~\eqref{eq:unob} is obtained from the definition.
We can rewrite
\begin{equation}
[v_0^T~v_1^T]\begin{bmatrix}
z_0I-A_{11} & -A_{12} & B_1\\
-A_{21} & z_0I-A_{22} & B_2
\end{bmatrix}=0,
\end{equation}
which means that $[A,~B]$ is not controllable.$\hfill \square$
\end{proof}

\section{Appendix: proof of Proposition~\ref{prop:convert}}\label{se:AppB}
The proof of Proposition~\ref{prop:convert} will be divided into several steps. Start by rewriting eq.~(\ref{eq:Rrealization}) as
\begin{equation}
[sI-W,V]= (s{I}-{R})[I-{Q},{P}].
\end{equation}
For any $N\in\mathcal{E}_p$, and corresponding $R\in\mathcal{D}_p$, and any $[Q,P]$, define
\begin{equation}\label{eq:xy}
[X,Y]\triangleq [I,0]-N[I-Q,P].
\end{equation}
Let $q\le p$ be the normal row rank of $[X,Y]$ and consider the Smith-McMillan of $[X,Y]=U(s)M(s)V(s)$,
where $U(s),V(s)$ are unimodular matrices in $s$, and  
$$M(s)=\begin{bmatrix} 
\frac{\alpha_{1}(s)}{\beta_{1}(s)} & 0 & 0 & \cdots  & 0 \\
0 & \frac{\alpha_{2}(s)}{\beta_{2}(s)} & 0 & \cdots & 0\\
\vdots &  & \ddots & & \vdots\\
0 & \cdots & 0 & \frac{\alpha_{q}(s)}{\beta_{q}(s)} & 0 \\
0 & \cdots & 0 &  \cdots & 0 
\end{bmatrix},$$
where $\alpha_i$ divides $\alpha_{i+1}$ and $\beta_i$ divides $\beta_{i+1}$ for any $i=1,2,\ldots,q-1$. Let $j=1,\ldots,q$ be the smallest integer that $s|\beta_j(s)$, which means that polynomial $s$ exactly divides polynomial $\beta_j(s)$ (otherwise $j=q+1$).

\begin{lemma}\label{lemma1}
If follows that
\begin{align}
&\min_{R\in\mathcal{D}_p} \text{deg}\left\{(s{I}-{R})s^{-1}[s{Q},s{P}]+[{R},{0}]\right\}\nonumber\\
&=\min_{N\in\mathcal{E}_p}  \left \{\text{deg} \left\{N[I-{Q},{P}]\right\}-q+j-1 \right \}.\label{eq:B3}
\end{align}
\end{lemma}

\begin{proof}
Rewrite eq.~(\ref{eq:Rrealization}) as
\begin{equation}
[sI-W,V]= (s{I}-{R})[I-{Q},{P}].
\end{equation}
Then, for any $N\in\mathcal{E}_p$, and corresponding $R\in\mathcal{D}_p$, and any $[Q,P]$,
\begin{align*}
[I-X,Y]=N[I-Q,P] & =(I-R/s)[I-Q,P]\\  & =[I-W/s,V/s]. 
\end{align*}
It follows that 
\begin{align*}
\text{deg}\{[W,V]\}=\text{deg}\{[sX,sY]\}=\text{deg}\{s[X,Y]\}.
\end{align*}
Rewrite the expression of $[sX,sY]$ as
\begin{align*}
[sX,sY]&=sU(s)M(s)V(s)
\triangleq U(s)M'(s)V(s),
\end{align*}
where $$M'(s)=sM(s)=\begin{bmatrix} 
\frac{s\alpha_{1}(s)}{\beta_{1}(s)} & 0 & 0 & \cdots  & 0 \\
0 & \frac{s\alpha_{2}(s)}{\beta_{2}(s)} & 0 & \cdots & 0\\
\vdots &  & \ddots & & \vdots\\
0 & \cdots & 0 & \frac{s\alpha_{q}(s)}{\beta_{q}(s)} & 0 \\
0 & \cdots & 0 &  \cdots & 0 
\end{bmatrix}.$$ 
Since 
$$\text{deg}[X,Y]=\sum_{i=1}^q\text{deg}[\beta_i(s)],$$
and $\alpha_i|\alpha_{i+1}$ for all $i$
then
\begin{align}
\text{deg}[sX,sY]&=\sum_{i=1}^q\text{deg}[\beta_i(s)]-q+j-1\nonumber\\
&=\text{deg}[X,Y]-q+j-1\nonumber\\
&=\text{deg}[I-X,Y]-q+j-1\nonumber \\
&=\text{deg}N[I-Q,P]-q+j-1. \label{eq:jj}
\end{align}
%\begin{remark}
%It follows that $j-1\le q\le p$.
%\end{remark}
%It is easy to see that eq.~\eqref{eq:deg} holds if and only if $\alpha_i(s)$ does not 
%process any zeros at $0$, or equivalently $\alpha_i(z)$ and $z$ are coprime. 
%From the property of Smith-McMillam form,  $\alpha_i(z)$ and $z$ are coprime is equivalent to that there are no zeros of $G(z)$ at $0$.
%\subsection{$G$ does not have full normal row rank}
%When $G(z)$ does not have full normal row rank $P$ but $Q<P$,  then the $M(z)$ in the Smith-McMillan form has the following form
%$$M_1(z)=\begin{bmatrix} 
%\frac{\alpha_{1}(z)}{\beta_{1}(z)} & 0 & 0 & \cdots  & 0 \\
%0 & \frac{\alpha_{2}(z)}{\beta_{2}(z)} & 0 & \cdots & 0\\
%\vdots &  & \ddots & & \vdots\\
%0 & \cdots & 0 & \frac{\alpha_{Q}(z)}{\beta_{Q}(z)} & 0 \\
%0 & \cdots & 0 &  \cdots & 0 
%\end{bmatrix}$$
%After left multiply $\Lambda(z)$
%\begin{align*}
%H(z)&=\Lambda(z)G(z)=z^{-L'}U(z)M_1(z)V(z)\\
%&=U(z)z^{-L'}M_1(z)V(z)=U(z)z^{-L'}\begin{bmatrix} I_Q & 0 \\ 0 & 0 \end{bmatrix}M_1(z)V(z)\\& \triangleq U(z)M_1'(z)V(z),
%\end{align*}
%it is then easy to deduce that $z^{-L'}G(z)$ has $$\text{deg} [H(z)] = L'\times Q + \text{deg} [G(z)].$$
%
Based on the above analysis, we can reformulate the optimisation on the left-hand side of eq.~\eqref{eq:B3} into the following form
\begin{align}
\min_{N\in\mathcal{E}_p} \text{deg}[W,V]&=\min_{N\in\mathcal{E}_p} \text{deg}[sX,sY]\nonumber\\
&=\min_{N\in\mathcal{E}_p}\{\text{deg}N[I-Q,P]-q+j-1\}.
\label{eq:N}
\end{align}
$\hfill \square$
\end{proof}

The above optimisation is hard to solve since both $\text{deg}N[I-Q,P]$, $q$ and $j$ depend on the choice of $N$. However, we will show next that an $N^*\in\mathcal{E}_p$ that minimises $\text{deg}N[I-Q,P]$ is also a solution to eq.~\eqref{eq:N}. Such $N^*\in\mathcal{E}_p$ results in $j=1$ and $q=p$. 

The remaining part of the proof will use notation and content from sections~\ref{se:4.1} and ~\ref{se:4.2}. Hence, the reader is expected to have read these sections before continuing.  First, we shall discuss why an optimal $N^*$ guarantees $j=1$ and then that it also guarantees $q=p$. From eqs.~(\ref{eq:Nnew1}) and~(\ref{eq:Nnew2})
\begin{align}
[I-X,Y]&=\text{diag}\left\{\frac{sd_1-n_1}{sd_1},\ldots,\frac{sd_p-n_p}{sd_p}\right\}[I-Q,P]\nonumber\\ &=\frac{1}{s}\text{diag}\left\{\frac{\hat{n}_1}{d_1},\ldots,\frac{\hat{n}_p}{d_p}\right\}[I-Q,P]
\label{eq:guarantee}
\end{align}
%\begin{lemma}\label{lemma:A2}
%Since $[W,V]=sI-\text{diag}\{\frac{\hat{n}_1}{d_1},\ldots,\frac{\hat{n}_p}{d_p}\}[I-Q,P]$. If $[W,V]$ does not process any zeros at $0$, then $j=1$.
%\end{lemma}
%\begin{proof}
%This can be similarly shown using the Smith-McMillan form of $[W,V]$.$\hfill \square$
%\end{proof}

\begin{lemma}
With $N^*$ defined in~(\ref{eq:Np}), $j=1$ in eq.~(\ref{eq:jj}).
\end{lemma}
\begin{proof}
Let\begin{align}
[X^*,Y^*]&= I-\frac{1}{s}\text{diag}\left\{\frac{\hat{n}^*_1}{d^*_1},\ldots,\frac{\hat{n}^*_p}{d^*_p}\right\}[I-Q,P]\\
&=\frac{1}{s}\left(sI-\text{diag}\left\{\frac{\hat{n}^*_1}{d^*_1},\ldots,\frac{\hat{n}^*_p}{d^*_p}\right\}[I-Q,P]\right)\label{eq:xystar}
\end{align}
From the design process, if $[I-Q,P]$ has a zero at $0$, then it would be cancelled by designing $N^*$. $N^*$ is not designed have a zero at $0$ unless it was used to cancel poles in $0$ of $[I-Q,P]$. Then $N^*[I-Q,P]=\frac{1}{s}\text{diag}\left\{\frac{\hat{n}^*_1}{d^*_1},\ldots,\frac{\hat{n}^*_p}{d^*_p}\right\}[I-Q,P]$ in which $\text{diag}\left\{\frac{\hat{n}^*_1}{d^*_1},\ldots,\frac{\hat{n}^*_p}{d^*_p}\right\}[I-Q,P]$ does not have a zero at $0$.

Next, we show that $sI-\text{diag}\{\frac{\hat{n}^*_1}{d^*_1},\ldots,\frac{\hat{n}^*_p}{d^*_p}\}[I-Q,P]$ does not have a zero at $0$. Otherwise, there would exist a $v^T\in\mathcal{R}^{1\times (p+m)}$ such that 
\begin{align}
&\left. v^T\left(sI-\text{diag}\left\{\frac{\hat{n}^*_1}{d^*_1},\ldots,\frac{\hat{n}^*_p}{d^*_p}\right\}[I-Q,P]\right)\right\vert_{s=0}=0\\
&\left. \Rightarrow v^T\left(\left\{\text{diag}\frac{\hat{n}^*_1}{d^*_1},\ldots,\frac{\hat{n}^*_p}{d^*_p}\right\}[I-Q,P]\right)\right\vert_{s=0}=0,
\end{align}
which would lead to $\text{diag}\left\{\frac{\hat{n}^*_1}{d^*_1},\ldots,\frac{\hat{n}^*_p}{d^*_p}\right\}[I-Q,P]$ having a zero at $0$ and a contradiction. Then $[X^*,Y^*]$ does not have any zero at $0$, and therefore $j=1$ from a similar analysis using Smith-McMillan form.
$\hfill \square$
\end{proof}

\begin{lemma}
With $N^*$ defined in~(\ref{eq:Np}), $q=p$ in eq.~(\ref{eq:jj}).
\end{lemma}
\begin{proof}
%We first show that $[I-Q,P]$ has full normal row rank since $Q$ is a strict proper transfer matrix. Then for any $N=(I-R/s)$, $[I-X,Y]=N[I-Q,P]$ ha full normal row rank. \yy{this can be shown by contradiction.} 
%
%For any $N$, if $[X,Y]=[I,0]-N[I-Q,P]$ does not have full normal rank, then there exists $v^T$ such that
%\begin{align}
%&v^T(I-N(I-Q))=0\Rightarrow v^T=v^TN(I-Q)\\
% &v^TNP=0.
% \end{align}
%Since $I-Q$ has full normal rank, we have $v^TN=v^T(I-Q)^{-1}$
% substitute it to the second equation, 
% \begin{equation}
% v^TG=0.
% \end{equation}
From eq.~\eqref{eq:xystar}, to show that $[X^*,Y^*]$ has full normal rank, it is equivalent to show that 
$s[X^*,Y^*]= sI-\text{diag}\left\{\frac{\hat{n}^*_1}{d^*_1},\ldots,\frac{\hat{n}^*_p}{d^*_p}\right\}[I-Q,P]$ has a full normal rank. 

Since $\text{diag}\left\{\frac{\hat{n}^*_1}{d^*_1},\ldots,\frac{\hat{n}^*_p}{d^*_p}\right\}[I-Q,P]$ does not have a zero at $0$, then it has full rank. At $s=0$, we have that
\begin{equation}
\left. \text{rank}\left(\text{diag}\left\{\frac{\hat{n}^*_1}{d^*_1},\ldots,\frac{\hat{n}^*_p}{d^*_p}\right\}[I-Q,P]\right)\right\vert_{s=0}=p.
\end{equation}
Hence, the normal rank of $s[X^*,Y^*]=p$. Therefore, with $N^*$ defined in~(\ref{eq:Np}), $q=p$ in eq.~(\ref{eq:jj}).$\hfill \square$
\end{proof}

In summary, $N^*$ obtained in Algorithm~\ref{alg:iqpzero} minimises not only eq.~\eqref{eq:Np} but also eq.~\eqref{eq:N}. This completes the proof of Proposition~\ref{prop:convert} since  $N^*$, $j=1$ and $q=p$ minimise eq.~\eqref{eq:N}.
\end{document}